\documentclass[conference]{IEEEtran}
\usepackage{graphicx}
\usepackage{amsmath, amssymb, amsthm, amscd, MnSymbol}
\usepackage{mathtools}
\usepackage{marginnote}
\usepackage[linesnumbered,lined,boxed,commentsnumbered]{algorithm2e}

\usepackage{cite}
\newtheorem{theorem}{Theorem}[section]
\newtheorem{lemma}[theorem]{Lemma}

\def\H{{\bf H}}
\def\W{{\bf W}}
\def\X{{\bf X}}
\def\S{{\bf S}}

\def\R{{\bf R}}

\begin{document}

\title{Efficient Optimal Joint Channel Estimation and Data Detection for Massive MIMO Systems}

\author{\IEEEauthorblockN{Haider~Ali~Jasim~Alshamary}
\IEEEauthorblockA{Department of ECE, University of Iowa, USA}
\and
\IEEEauthorblockN{Weiyu~Xu}
\IEEEauthorblockA{Department of ECE, University of Iowa, USA}}
%\and
%\IEEEauthorblockN{James Kirk\\ and Montgomery Scott}
%\IEEEauthorblockA{Starfleet Academy\\
%San Francisco, California 96678--2391\\
%Telephone: (800) 555--1212\\
%Fax: (888) 555--1212}}
\maketitle
\begin{abstract}
In this paper,  we propose an efficient optimal joint  channel estimation and data detection algorithm for massive MIMO wireless systems. Our algorithm is optimal in terms of the generalized likelihood ratio test (GLRT).  For massive MIMO systems, we show that the expected complexity of our algorithm grows polynomially in the channel coherence time. Simulation results demonstrate significant performance gains of our algorithm compared with suboptimal non-coherent detection algorithms. To the best of our knowledge, this is the first algorithm which efficiently achieves GLRT-optimal non-coherent detections for massive MIMO systems with general constellations.
\end{abstract}

\IEEEpeerreviewmaketitle

\section{Introduction}
\label{sec:intro}

Massive MIMO wireless systems emerge as an important potential technology for next generation wireless communications. Massive MIMO systems aim to meet the growing demand for higher data rates and wider coverage, by using hundreds of antennas at base stations (BS). In \cite{Thomas}, the author showed that, in a single-cell massive MIMO system, when the number of receive antennas $N \rightarrow \infty$, we can eliminate the negative effects of fast fading and non-correlated noise. Massive MIMO systems can also greatly boost the energy efficiency of cellular wireless communications \cite{Improving_energy}.
Besides larger capacity and higher energy efficiency, massive MIMO systems can boost  the robustness both to unintended man-made interference and to intentional jamming. These advantages make massive MIMO a promising candidate for new 5G wireless communications technologies \cite{NextGeneration,Jeon}.

In this paper, we consider a typical TDD (Time Division Duplexing) massive MIMO wireless systems. In TDD massive MIMO systems, user terminals equipped with single antennas transmit pilot sequences and information data to base stations in uplink transmissions. Exploiting channel reciprocity in TDD systems, the base stations use channel estimation from uplink transmissions for precoding in downlink data transmissions. One of the most prominent challenges in massive MIMO systems is timely acquiring the channel state information (CSI) for a large number of antenna pairs. In fact, unknown CSI is a bottleneck in achieving the full potentials of massive MIMO systems \cite{NextGeneration}. Especially in fast fading environments where the channels change rapidly,  one would need to dedicate a significant portion of the coherence interval for pilot sequences, leaving few times slots for data transmissions. In multi-cell multi-user  massive MIMO systems, due to the scarcity of resources for orthogonal pilot sequences, pilot sequences from neighboring cells will inevitably pollute channel estimations in the current cell, causing the issue of pilot contamination \cite{ MUMIMO,HowmanyAntenna,NextGeneration}.  With the need of obtaining good CSI, pilot contaminations fundamentally limit the achievable data rates in massive MIMO systems.

It is known that joint channel estimation and data detections can greatly alleviate the issue of pilot contamination, and enhance system performance for massive MIMO systems \cite{Wilson, NextGeneration, Muller, Ngo, Joint_channel}. Maximum likelihood (ML) or GLRT-optimal joint channel estimation and data detection algorithms are especially attractive due to their optimality, when the channel statistics are known or unknown.  However, existing efficient joint channel estimation and data detection algorithms for massive MIMO systems are suboptimal, and cannot achieve the optimal non-coherent detection performance.  It is thus of great interest to design efficient optimal non-coherent data detection algorithms for massive MIMO. Moreover, it is important to obtain the  performance limits of ML or GLRT-optimal non-coherent data detection such that they can be used as benchmarks for evaluating low-complexity suboptimal non-coherent data detection algorithms.

For special cases of conventional MIMO systems with few antennas,  there exist a few efficient algorithms which can achieve optimal joint channel estimation and data detection.  In \cite{Hassibi, Ma, Optimal}, the authors introduced sphere decoders to achieve the exact ML non-coherent signal detection for constant-modulus constellations (such as BPSK, QPSK). The sphere decoder demonstrates low computational complexity for high signal-to-noise ratio (SNR) and moderate system dimensions \cite{complexity}. However, the sphere decoder has  exponential complexity at low or constant SNR \cite{complexity2}, and it only works for non-coherent single input multiple output (SIMO) systems.
%However, blind or semiblind and pilot base algorithms \cite{Stoica,Ma}. However, most the efficient non-coherent signal detection algorithms for conventional MIMO systems are suboptimal with limited exceptions like the sphere decoder.
%Generalized-likelihood-ratio-test (GLRT) is utilized in \cite{GLRT} for non-coherent detection,  attaining cubic complexity in channel coherence time $T$.????)

Another line of works use the method of channel state space partition for ML or  GLRT-optimal non-coherent data detection,  attaining polynomial complexity in channel coherence time $T$ \cite{Anastasopoulos,GLRT,Madhow,AAaproach}. For example, \cite{AAaproach} achieves GLRT-optimal non-coherent detection for  pulse-amplitude modulation (PAM) using auxiliary-angle approach of polynomial-time complexity. However, these works are only for single-input single-output systems where the channel coefficient is a single complex variable, or orthogonal space time block coded MIMO systems with a single receive antenna. These algorithms using channel state partition do not work efficiently for general MIMO systems with a few more receive antennas,  not to mention massive MIMO systems with hundreds of receive antennas.

In  \cite{Weiyu}, the authors  proposed an efficient branch-estimate-and-bound algorithm for GLRT-optimal non-coherent data detection for conventional MIMO systems with general constellations.  Although this algorithm from \cite{Weiyu} was the only known efficient GLRT-optimal algorithm for general MIMO systems with multiple transmit and receive antennas and general constellations, this approach does not have a computational complexity polynomial in channel coherence time, and is only for conventional MIMO systems with few receive antennas.

 In this paper, we propose a novel efficient GLRT-optimal joint channel estimation and data detection for massive MIMO systems with general constellations. We show that our algorithm has an expected computational complexity  polynomial in the channel coherence time $T$ for massive MIMO systems.  In its essence, our approach is a branch-and-bound  method on the residual energy of massive MIMO signals after projecting them onto certain subspaces.   To the best of our knowledge, this framework is the first GLRT-optimal non-coherent signal detection algorithm for massive MIMO systems with low computational complexity and optimal performance.  Moreover, our algorithm can provide benchmark performance against which we can evaluate suboptimal low-complexity joint channel estimation and data detection algorithms.

The rest of this paper is organized as follows. Section \ref{sec:problem} presents the system model. In Section \ref{sec:algorithm}, we introduce the new GLRT-optimal non-coherent data detection algorithm. The expected complexity of the algorithm is derived in Section \ref{sec:complexity}. Section \ref{sec:Results} demonstrates the empirical performance and the computational complexity of ML algorithm. %Section \ref{sec:Conclusion} concludes this paper. 
\section{Joint Channel Estimation and Signal Detection for Massive MIMO}
\label{sec:problem}
We consider a TDD massive MIMO wireless system with $N$ receive antennas at the base station, and $M\ll N$ user terminals each equipped with a single antenna. We assume a discrete-time block flat fading channel model where the channel coefficients are fixed for a coherence time $T$.  Across different fading blocks, the channel coefficients take independent values from unknown distributions.  We model the uplink transmission of this system within one channel block by
\begin{equation}
\X=\H \S^*+\W, \label{eq:system}
\vspace{-0.04 in}
\end{equation}
%\vspace{-0.05 in}
where $\X \in \mathcal{C}^{N \times T}$ is the received signal at the BS, $\S^*$ is an $ M \times T$ matrix representing the transmitted signal, whose entries are independent and identically distributed (i.i.d.) symbols from a modulation constellation $\Omega$ ($\Omega$ can be of constant or non-constant modulus, such as 16-QAM), $\W \in \mathcal{C}^{N \times T}$ represents additive noises, and $\H \in \mathcal{C}^{N \times M}$ represents the unknown channel matrix. The elements of $\W$ are i.i.d. random variables following circularly symmetric complex Gaussian distribution $\mathcal{N}(0,\sigma_w^2)$. In each channel coherence block, we further assume that the channel coefficients are deterministic with no prior statistical information known about them \cite{Stoica,Ma}.
%and $\H$ are assumed to be i.i.d. complex Gaussian random variables

Since the channel coefficients take unknown deterministic values, we can formulate the GLRT-optimal joint channel estimation and data detection as a mixed optimization problem over $\H$ and $\S$: %\vspace{-0.05in}
\begin{equation}
\min_{\H,\S^* \in {\Omega}^{M \times T}}\| \X-\H\S^*\|^2,
\label{eq:mixed}
\vspace{-0.04 in}
\end{equation}
%\vspace{-0.05 in}
where $\Omega^{M \times T}$ represents the signal lattice of dimension ${M \times T}$.  We remark that the GLRT-optimal detection is equivalent to ML detection for SIMO systems with constant-modulus modulations,  and for MIMO systems with equal-energy signaling, when the channel coefficients are known to take i.i.d. circularly symmetric complex Gaussian values \cite{Madhow}.

We note that the combinatorial optimization problem in (\ref{eq:mixed}) is a least squares problem in $\H$ , while an integer least-squares problem in $\S^*$, since each element of $\S^*$ is chosen from a discrete constellation $\Omega$ \cite{Hassibi}. Hereby, for any given $\S^*$, the channel matrix $\H$ that minimizes (\ref{eq:mixed}) is given by $\hat{\H}=\X(\S^*)^{\dag}$,
%\begin{equation}
%\hat{\H}=\X(\S^*)^{\dag}, \nonumber \label{eq:opth}
%\end{equation}
where $(\cdot)^{\dag}$ denotes the Moore-Penrose pseudoinverse of a matrix. Since $(\S^*)^{\dag}=\S(\S^*\S)^{\dag}$, $\hat{\H}=\X\S(\S^*\S)^{\dag}$. Substituting this into (\ref{eq:mixed}), we get
\begin{align}
\min_{\H,\S^*}\| \X-&\H\S^*\|^2 = \min_{\S^* \in \Omega^{M \times T}} \| \X( {\bf{I}}-\S(\S^*\S)^{\dag}\S^*) \|^2 \notag\\
&=\min_{\S^*} \text{tr}(\X( {\bf{I}}-\S(\S^*\S)^{\dag}\S^*) \X^*) \notag\\
&=\text{tr}(\X\X^*)-\max_{\S^* \in \Omega^{M \times T}}\text{tr}((\S^*\S)^{\dag}\S^*\X^*\X\S), %\notag\\
\label{eq:optmetric}
%\vspace{-0.1 in}
\end{align}
where $\text{tr}(\cdot)$ is the trace of a matrix. To simplify the mathematical formulation, we define $\Xi$ to be a new $M$-dimensional constellation, each element of which is an $M$-dimensional  vector with its entries taking values from $\Omega$. So the cardinality of $ \Xi$ is $|\Omega|^M$. Then we can rewrite (\ref{eq:optmetric}) as
\begin{equation}\label{eq:optmetric2}
\text{tr}(\X^*\X)-\max_{\S^* \in \Xi^{1 \times T}}\text{tr}((\S^*\S)^{\dag}\S^*\X^*\X\S),
\end{equation}
where we use $\text{tr}(\X^*\X)=\text{tr}(\X\X^*)$. Now by choosing $\rho_{min}$ to be the minimum eigenvalue of $\X^*\X$, the minimization problem in (\ref{eq:optmetric2}) can be equivalently represented by the following optimization problem,
\begin{equation}
\text{tr}(\X^*\X-\rho_{min}I)-\max_{\S^* \in \Xi^{1 \times T}}% \in \Xi^{1 \times T}}
\text{tr}((\S^*\S)^{\dag}\S^*(\X^*\X-\rho_{min}I)\S), \label{eq:max2}
\end{equation}
because $\text{tr} ((\S^*\S)^{\dag}\S^*(\rho_{min}I)\S)$ is a constant. Since $A=\X^*\X-\rho_{min}I$ is positive semidefinite, we can factorize  $A=\R^*\R$ using Cholesky decomposition, where $\R^*$ is the lower triangular matrix of Cholesky decomposition. Finally, using the trace property for product of matrices, (\ref{eq:max2}) can be transformed as follows:
\begin{align}
&\text{tr}({\R^*\R})-\max_{\S^* \in \Xi^{1 \times T}}\text{tr}((\S^*\S)^{\dag}\S^*{\R^*\R}\S) \notag\\
&=\min_{\S^* \in \Xi^{1 \times T}} \text{tr}(\R( {\bf{I}}-\S(\S^*\S)^{\dag}\S^*) \R^*) \notag\\
&=\min_{\S^* \in \Xi^{1 \times T}}||\R^*-\S(\S^*\S)^{\dag}\S^*\R^*||^2. \label{eq:min}
\end{align}
Thus our goal is to  minimize (\ref{eq:min}), based on which our novel algorithm is built. We remark that this approach of transforming the GLRT-optimal detection to (\ref{eq:min}) is novel, very different from existing approaches for GLRT-optimal detection including the sphere decoder \cite{Hassibi} which only works for SIMO wireless systems. We also note that, the channel estimate $\hat{\H}=\X(\S^*)^{\dag}$ can be used for downlink precoding  after solving (\ref{eq:min}).

\section{Efficient GLRT-Optimal  Joint Channel Estimation and Data Detection Algorithm}
\label{sec:algorithm}
Finding the optimal solution to (\ref{eq:min}) is a formidable task, since it requires searching over all the $|\Omega|^{MT}$ hypotheses in the signal space. The exhaustive search appproach  provides the optimal solution, however, its complexity grows exponentially in the channel coherence time. In the special case of SIMO systems, the sphere decoder efficiently solves GLRT-optimal detection (in a different format from (\ref{eq:min})) for both constant-modulus\cite{Hassibi} and nonconstant-modulus constellations \cite{Optimal}. However, the sphere decoders from \cite{Hassibi} and \cite{Optimal} do not work for MIMO systems.

To describe our algorithm, we first introduce a  tree representation of the signal space. Recall that we use $\Xi$ to represent the set of signal vectors of length $M$, where each element of each vector takes value from the constellation $\Omega$. We can thus represent the set of possible matrices for ${\S}^*$ by a tree of $T$ layers. At a layer $0$, we have one root node.   Each tree node at layer $i$, $0\leq i \leq (T-1)$, has $|\Xi|=|\Omega|^{M}$ child nodes. We use $\S^*_{1:i}$ to represent the first $i$ columns of $\S^{*}$, and each possible matrix for $\S^*_{1:i}$ corresponds to a layer-$i$ tree node.  And we call the tree nodes at layer $T$ as leaf nodes, and thus each possible matrix for ${\S}^*$ is represented by a leaf node.
Furthermore, for each possible matrix value for $\S^*$, we define its metric by
\begin{equation}
M_{{\S}^*}=||\R^*-{\S}({\S}^*{\S})^{\dag}{\S}^*\R^*||^2. \label{eq:MetricCom2}
\end{equation}
For a partial matrix $\S^*_{1:i}$, we define its metric by
\begin{equation}
M_{\S^*_{1:i}}=||\R_{1:i}^*-\S_{1:i}(\S^*_{1:i}\S_{1:i})^{\dag}\S^*_{1:i}\R_{1:i}^*||^2 , \label{eq:MetricCom}
\end{equation}
where $1\leq i \leq T$, and $\R^*_{1:i}$ is the first $i$ rows of $\R^*$. Thus solving (\ref{eq:min}) is equivalent to finding an $\hat{\S}^*$ that minimizes $M_{\S^*}$ among all the possible matrix values for $\S^*$.  

To develop our algorithm, we have the following lemma about the comparison between $M_{\S^*_{1:i}}$ and $M_{\S^*}$.
\begin{lemma}
For every $i\leq T$ and any matrix value for $\S^*$,   $M_{\S^*_{1:i}} \leq M_{\S^*}$
\label{lemma:comparison}
\end{lemma}
\begin{proof}
 We observe that $M_{\S^*}$ is the residual energy  after projecting the columns of $\R^*$ onto the subspace spanned by the columns of $\S$; and $M_{\S^*_{1:i}}$ is the residual energy  after projecting the columns of $\R^*_{1:i}$  (the first $i$ rows of $\R^*$)  onto the subspace spanned by the columns of $\S_{1:i}$ ( $\S_{1:i}$ is the just the first $i$ rows of $\S$). Since orthogonal linear projections minimize the residual energy among all linear projections, we can show,  at the first $i$ indices,   the residual energy $M_{\S^*_{1:i}}$ after applying orthogonal  projections $\S_{1:i}(\S^*_{1:i}\S_{1:i})^{\dag}\S^*_{1:i}$ to $\R^*_{1:i}$, will be no bigger than these indices' residual energy (denoted by $Q$) after applying ${\S}({\S}^*{\S})^{\dag}{\S}$ to $\R^*$.  Moreover, for the orthogonal projection ${\S}({\S}^*{\S})^{\dag}{\S}$ applied to $\R^*$,  the total residual energy  $M_{\S^*}$ over $T$ indices is no smaller than the residual energy $Q$ over the first $i$ indices. Because   $M_{\S^*} \geq Q$ and $Q\geq M_{\S^*_{1:i}}$, we have  $M_{\S^*_{1:i}} \leq M_{\S^*}$.
\end{proof}
Lemma \ref{lemma:comparison} means that $M_{\S^*_{1:i}}$ is a lower bound on $M_{\S^*}$. Intuitively, suppose that $M_{\S^*_{1:i}}$ is too big, then $M_{\S^*}$ must also be big, and $\S^*$ will not minimize (\ref{eq:min}).  This motivates us to propose the following branch-and-bound algorithm for GLRT-optimal joint channel estimation and data detection. In this algorithm, we set a search radius $r$, and use this radius $r$ to regulate a depth-first search over the signal tree structure for the optimal solution to (\ref{eq:min}). In fact, if $M_{\S^*_{1:i}}>r^2$, this algorithm will not search among the child nodes of $\S^*_{1:i}$. If the optimal solution is not found under the current radius $r$, we will increase the search radius $r$ for new searches until the optimal solution is found. The description of our algorithm is given in Algorithm \ref{algo_disjdecomp}.
\label{sec:Algorithm}
%\IncMargin{1em}
\begin{algorithm}[hb]
\SetKwData{Left}{left}\SetKwData{This}{this}\SetKwData{Up}{up}
\SetKwFunction{Union}{Union}\SetKwFunction{FindCompress}{FindCompress}
\SetKwInOut{Input}{input}\SetKwInOut{Output}{output}
\Input{radius $r$, matrix $\R$, constellation $\Xi$ and a $1 \times T$
index vector $I$}
\Output{The transmitted signal $\S^*$}
\BlankLine
%\emph{special treatment of the first line}\;
 Set $i=1$, $I(i)=1$ and set $\S^*_{1:i}=\Xi(I(i))$.

(Computing the bounds) Compute the metric $M_{\S^*_{1:i}}$. If
$M_{\S^*_{1:i}}>r^2$, go to 3; else, go to 4;

(Backtracking) Find the smallest $1 \leq j \leq i$ such
that $I(j)<|\Xi|$. If there exists such $j$, set $i=j$ and go to
5; else go to 6.

If $i=T$, store current $\S^*$, update $r^2=M_{\S^*_{1:T}}$ and go to 3; else set $i=i+1$, $I(i)=1$ and
$\S^*_{1:i}=\Xi(I(i))$, go to 2.

Set $I(i)=I(i)+1$ and $S^*_{i}=\Xi(I(i))$.
Go to 2.
If any sequence $\S^*$ is ever found in Step 4, output the latest
stored full-length sequence as the ML solution; otherwise, double $r$
and go to 1.
%\vspace{0.05 in}
\caption{ML channel estimation and signal detection algorithm.}
\label{algo_disjdecomp}
\end{algorithm} %\DecMargin{1em}
%\vspace{-0.051 in}
%In order to simplify the complexity analysis, we further modify step 6 of the ML algorithm: ``If any sequence $\s^*$ is ever found in step 4, the output of the latest stored full-length sequence will be the ML solution; otherwise, let $r=\infty$ and go to step 1''. We call such decoder as ``modified sphere decoder''. We emphasize that this change will not effect the optimality of the algorithm.
%\vspace{-0.1 in}
Algorithm \ref{algo_disjdecomp} is GLRT-optimal:
\begin{theorem}
Algorithm \ref{algo_disjdecomp} %in \ref{sec:Algorithm}
gives the optimal solution to (\ref{eq:min}).
\label{thm:ltt1}
%\vspace{-0.1 in}
\end{theorem}
This theorem is a result of Lemma \ref{lemma:comparison} and the branch-and-bound search over the signal space. 
\subsection{Metric calculation and initial radius $r$}
To compute $M_{\S^*_{1:i}}$, we can have a constant computational complexity independent of $T$,  by recursive calculations over tree structure.
The metric in (\ref{eq:MetricCom}) is equivalent to
\begin{equation}
M_{\S^*_{1:i}}=\text{tr}(\R^*_{1:i}\R_{1:i})-\text{tr}((\S^*_{1:i}\S_{1:i})^{\dag}\S^*_{1:i}{\R^*_{1:i}\R_{1:i}}\S_{1:i}). \label{eq:MetricCom3}
\end{equation}
From (\ref{eq:MetricCom3}), we can calculate the metric $M_{\S^*_{1:i}}$ efficiently. First, the term $\text{tr}(\R^*_{1:i}\R_{1:i})$ can be precomputed. Second, after defining a $T\times M$ matrix $A_{i}=\R_{1:i}\S_{1:i}$, we can update $A_{i+1}$ sequentially as $A_{i+1}=A_i+\R_{i+1:i+1}\S_{i+1:i+1}$. Similarly, we can define $M \times M$ matrix $B_i=\S^*_{1:i}\S_{1:i}$ and then sequentially update $B_{i+1}=B_{i}+\S^*_{i+1:i+1}\S_{i+1:i+1}$. Furthermore, the complexity of calculating $B^{\dag}_{i+1}$ is  $O(M^2)$ using matrix inversion lemma, where $B^{\dag}_{i+1}=(B_i+\S^*_{i+1:i+1}\S_{i+1:i+1})^{\dag}$.  The complexity of all these recursive updates do not depend on $T$ (noting that only $i$ rows of $A$ are nonzero).

For large $N$, we can choose the radius $r^2=cN $, where $c$ is any sufficiently small constant (please the next section for justifications). In fact, one can also use best-first tree search to find the optimal solution while  avoiding picking an $r$ beforehand. 
\section{Expected Computational Complexity}
\label{sec:complexity}

The computational complexity of our tree search based algorithms is mainly determined by the number of visited nodes in each layer. By ``visited nodes'', we mean the partial sequences $\S^*_{1:i}$ for which metric $M_{\S^*_{i}}$ is computed. The fewer the visited nodes, the lower computational complexity of our tree search algorithm has. {In this section, we show that the expected number of visited nodes  will grow linearly with $T$ under a sufficiently large number of receive antennas.}
To analyze the expected number of visited nodes, we assume that the channel coefficients are i.i.d. complex Gaussian random variables following distribution $\mathcal{N}(0,1)$. We also assume that the $M$ users send $M$ orthogonal pilot sequences between time indices 1 and $M$.
\begin{theorem}
Let $M$ be fixed, and let $r^2=c N$, where $c$ is any sufficiently small positive constant. Then for the tree search algorithm, the expected number of visited points at layer $i$  converges to $|\Xi|=|\Omega|^M$ for $i \geq (M+1 )$, as the number of receive antennas $N$ goes to infinity. The tree search algorithm only visits one tree node at each layer $i< (M+1)$.
\label{thm:complexity_fixedM}
\end{theorem}
\vspace{-0.2 in}
Due to space  limitations, we give an outline of the proof of Theorem \ref{thm:complexity_fixedM}.
\vspace{-0.09 in}
\begin{proof} (outline)
We first prove that, the tree search algorithm only visits $|\Xi|=|\Omega|^M$ nodes per layer when $\X^* \X= E[\X^* \X]$, where the expectation is taken over the distribution of channel coefficients.  Then we show that, when $N \rightarrow \infty$, $\X^* \X/N \rightarrow E[\X^* \X]/N$ in probability and that the expected number of visited nodes at layer $i$ ($(M+1)\leq i \leq T$)  approaches $|\Xi|$.

We first note that, the number of visited nodes at layer $i$ ($(M+1)\leq i \leq T$) is equal to $|\Xi|$, if there is one and only one sequence $\widetilde{\S}^*_{1:(i-1)}$ such that $M_{\widetilde{\S}^*_{1:(i-1)}} \leq r^2$. Let us consider the true transmitted sequence $\S^*$ . Then we have
\begin{align}\label{eq:Metric}
{E}&[\X^*\X] =E[(\H \S^*+\W)^*(\H \S^*+\W)]\notag \\
&= \S E[\H^*\H] \S^*+ E[\W^*\W]+ \S E[\H^* \W]+E[\W^*\H] \S^*\notag \\
&=N\S\S^*+N \sigma_{w}^2 I,
\end{align}
where the second equality is from $E[\H\H^*]=N I$ and $E[\H^* \W]=0$.
%when $i\neq j \Rightarrow E[\h_i^* \h_j]=0$
%\begin{equation}
%{E}[\X^*\X]/N=\S \S^*+\sigma^2_{w} I. \label{eq:Metric}
%\end{equation}
%%%%%%%%%%%%%%%%%%
%For the expected case, if we substitute (\ref{eq:Metric}) in (\ref{eq:optmetric}), then the minimization equation will be equal %to zero, which is the ideal case that can be achieved as $N \rightarrow\infty$. Under this assumption, if we set the initial %radius (i.e. the metric of the transmitted signal $\S^*$ ) to zero, then the metric of any other subsequences, rather than %those which emanate from the same branch of the transmitted signal $\S^*$, can not be equal to zero.

Because $\S\S^*$ is of rank $M$ with $M<T$, from (\ref{eq:Metric}), the minimum eigenvalue of ${E}[X^*X]/N$ is $\sigma_{w}^2$.  Then for the tree search algorithm (after scaling $A$ by a constant $N$),  $A={E}[X^*X]/N-\sigma^2_{w} I=\S\S^*$.  From the Cholesky decomposition, we know that $A=\S\S^*=\R^*\R$. This means that the columns of $\R^*$ span the same subspace as the columns of $\S$.  Thus the metric $M_{{\S}^*}=0$, because $||\R^*-{\S}({\S}^*{\S})^{\dag}{\S}^*\R^*||^2$  is precisely the residual energy after projecting the columns of $\R^*$ onto the subspace spanned by the columns of $\S$.   Since $M_{\S^*_{1:i}}\leq M_{{\S}^*}$,  $M_{\S^*_{1:i}}=0$ for all $i$.

Let us instead consider any signal matrix $\bar{\S}$ such that $ \bar{\S} \neq \S$ and $\bar{\S}^*_{1:M}=\S^*_{1:M}$ (namely $\bar{\S}$ shares the same pilot sequences as $\S$). For such $\bar{\S}$,  we can show that  $||\R^*-\bar{\S}(\bar{\S}^*\bar{\S})^{\dag}{\bar{\S}}^*\R^*||^2>0$, and that  $M_{\bar{\S}^*_{1:i}}> 0$ for the first $i$ such that $\bar{\S}^*_{1:i} \neq {\S}^*_{1:i}$.   In fact,  $M_{\bar{\S}^*_{1:i}}$ is no smaller than
$$D=\min_{i>M, \S, \bar{\S}, \S_{1:i}\neq \bar{\S}_{1:i}}  ||\S_{1:i}^*-\bar{\S}_{1:i}(\bar{\S}^*_{1:i}\bar{\S}_{1:i})^{\dag}{\bar{\S}_{1:i}}^*\S^*_{1:i} ||^2>0. $$
Thus for a search radius $r^2<D$, there will be only $T$ tree nodes (namely those from transmitted signal $\S^*$ ) with metric no bigger than $r^2$. This means that the tree search algorithm will visit at most $T |\Xi|$ tree nodes, under the assumption that $\X^* \X= E[\X^* \X]$.

For massive MIMO systems, when $N\rightarrow \infty$,  $\X^* \X/N \rightarrow  E[\X^* \X]/N$ in probability. In fact, we can show that, as $N \rightarrow \infty$,  with probability at least $(1-\epsilon)$, the tree search algorithm will visit at most  $T* |\Xi|$ tree nodes, where $\epsilon>0$ is an arbitrary small number.  With probability $\epsilon>0$, the tree search algorithm will visit at most $|\Xi|^T$ nodes. When $N \rightarrow \infty$,  $\epsilon$ can be pushed small fast enough such that the expected number of visited tree nodes grows linear in $T$.
\end{proof}
%We can think of $\overline{R}$ as a mapping of $M$ dimensional space matrix $\S^*$ onto $T$ dimensional space. Since the rank of matrix $A$ is $M$, consequently $\overline{R}$ has the same rank as well. With this setting, the projection of $\overline{R}$ away from $\S^*$ is nothing but $\overline{R}$ itself.
%%%%%%%%%%%%%%%%
%Now we can turn (\ref{eq:max}) to a minimization problem over the transmitted signal $\S^*$,
%\begin{equation}
%\min_{\S^* \in \Omega^{M \times T}} \overline{R}-P_{\overline{R}|\S^*}, \label{eq:metric}
%\end{equation}
 %where $P_{\overline{R}|\S^*}$ is the projection of $\overline{R}$ away from $\S^*$. The goal here is to find the the %transmitted sequence $\hat{\S}$ that minimize the $T \times T$ residual matrix $C=\overline{R}-P_{\overline{R}|\S^*}$ %which has been shown to be equal to zero as $N$ goes to infinity.
\vspace{-0.09 in}
Moreover, we only need $N$ to grow polynomially in $T$, in order to guarantee that the expected number of visited tree nodes grows polynomially in $T$. In fact, using large deviation bounds for the convergence of  $\X^* \X/N$  to $E[\X^* \X]/N$,  we have the following theorem.
\begin{theorem}
Let $M$ be fixed, and let $r^2=c N$, where $c$ is any sufficiently small positive constant. Then we only need the number of receive antennas $N$ to grow polynomially in $T$, to guarantee that  the expected number of visited points at layer $i$  converges to $|\Xi|$ for $i \geq (M+1 )$.
\label{thm:complexity_polynomialN}
\end{theorem} 
\section{Simulation Results}
\label{sec:Results}
In this section, we numerically simulate the performance of our new GLRT-optimal tree search algorithm, comparing it against suboptimal iterative and non-iterative MMSE channel estimation and data detection schemes.
We allow the receiver to know the first $M$ columns of the transmitted signal $\S^*$, which serve as necessary orthogonal pilot sequences o guarantee good error performance.  The non-iterative MMSE channel estimation scheme first uses the pilot sequences to perform MMSE channel estimation, and then uses the estimated channel to detect the transmitted information symbols. The iterative MMSE scheme iteratively exploits the detected data from the previous iteration to perform channel estimation used for data detection in the current iteration.

We consider different numbers of users $M$,  different number sof receive antennas $N$, and different block lengths $T$. We define the SNR as $\text{SNR}=\frac{ E||\H \S^*||^2}{E||\W||^2}$.   In Figure \ref{M2Perform}, we demonstrate the symbol error rate (SER) performance, as a function of SNR, for 16-QAM modulation.  Here, $M=2$, and $T=8$. For $10^{-2}$ SER and $N=50$,  our tree search algorithm provides $5$ dB gain over the iterative MMSE channel estimation scheme. When $N=100$, our method holds $6$ dB gain over the iterative MMSE scheme at $10^{-3}$ SER. Most importantly, our tree search algorithm guarantees providing the GLRT-optimal solution. Figure \ref{M2N200Perform} also demonstrates significant gains for $N=200$ antennas.

In Figure \ref{M2complexity} we evaluate the average number of visited nodes of the tree search algorithm.  Here $T=8$,  $M=2$, and the modulation scheme is 16-QAM . We observe tremendous reduction in the number of hypotheses that need to be tested,  compared with  exhaustive search method. For instance, for $N=100$ and $\text{SNR}=3$dB, exhaustive search requires testing $2.81\times 10^{14}$ hypotheses in each coherence block,  while our tree search algorithm visits only $5.5 \times 10^4$ tree nodes on average.  We remark that, using the same computer for simulation, exhaustive search would need $3.52 \times 10^3$ years to compute the optimal solution for one channel coherence block.
%The smallest possible number of visited points by our algorithm is $(T-M)\times |\Xi| \simeq 1500$. This number is reached by our tree search algorithm when $N=500$ and $SNR=6$dB.

%For constant modulus constellation, QPAK, and $M=4$, Figure \ref{M4Perform} shows the performance of the exact ML algorithm compared with iterative and non-iterative MMSE schemes. We choose block length to be $T=10$. ML achieves roughly 1 dB gain over iterative and non-iterative MMSE schemes for $N=50$, and $SER=10^{-4}$. For constant modulus, MMSE schemes show less gain loss compared with non-constant modulus in Figure \ref{M2Perform} for two reasons: First, the length of the embedded training sequences that the MMSE uses to estimate the channel is almost half of the coherence time. This helps the MMES to estimate the channel accurately, and we can easily realize that since estimating the channel iteratively does not attain better performance from that of non-iterative scheme; in fact, they are equivalent. Second, for QPAK constellation, it is easier for MMSE to make detection decision since it is based on phase difference between the 4 symbols of the QPAK, rather than the phase and amplitude for 16-QAM. Obviously for higher noise circumstances, higher SNR, or longer coherence time, the ML would express higher dB gain over non-optimal schemes.

We plot the average number of visited points as a function of SNR in Figure \ref{M4complexity}, for QPSK modulation, $M=4$, and $T=10$. We observe that increasing $N$ form $50$ to $500$ will greatly reduce the number of visited nodes.  Exhaustive search would need to examine $2.81\times 10^{14}$ hypotheses and will take $2000$ years to calculate the optimal solution for one channel coherence block. %
\begin{figure}[!t]%[!htb]
  \centering
 %\vspace{2.0cm}
  \includegraphics[width =2.5 in]{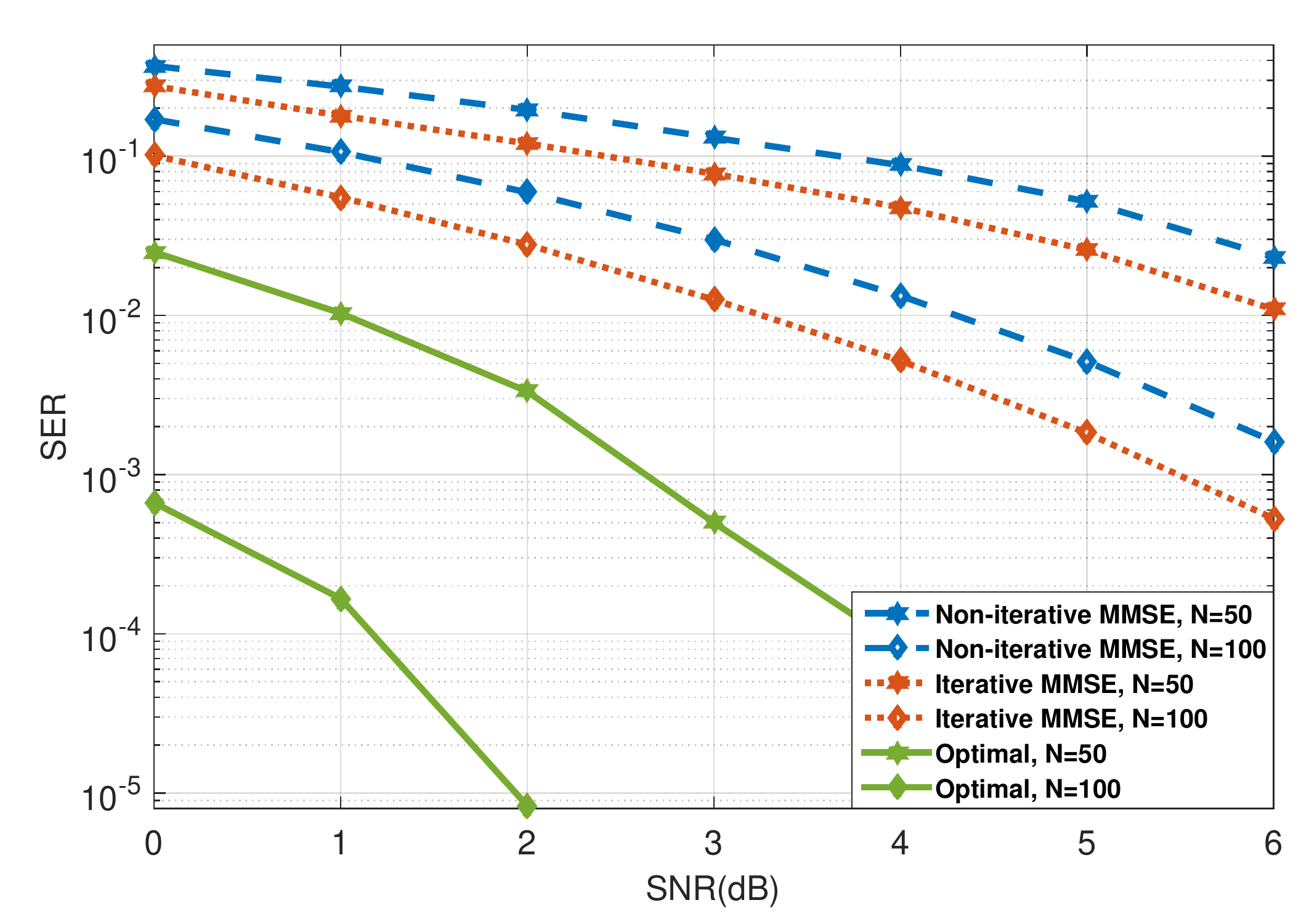}\\
  \caption{SER for iterative MMSE, non-iterative MMSE, and our optimal tree search algorithm. $M=2$, $T=8$, and 16-QAM constellation. }\label{M2Perform}
\end{figure}
\begin{figure}[!t]
  \centering
  \includegraphics[width =2.5 in]{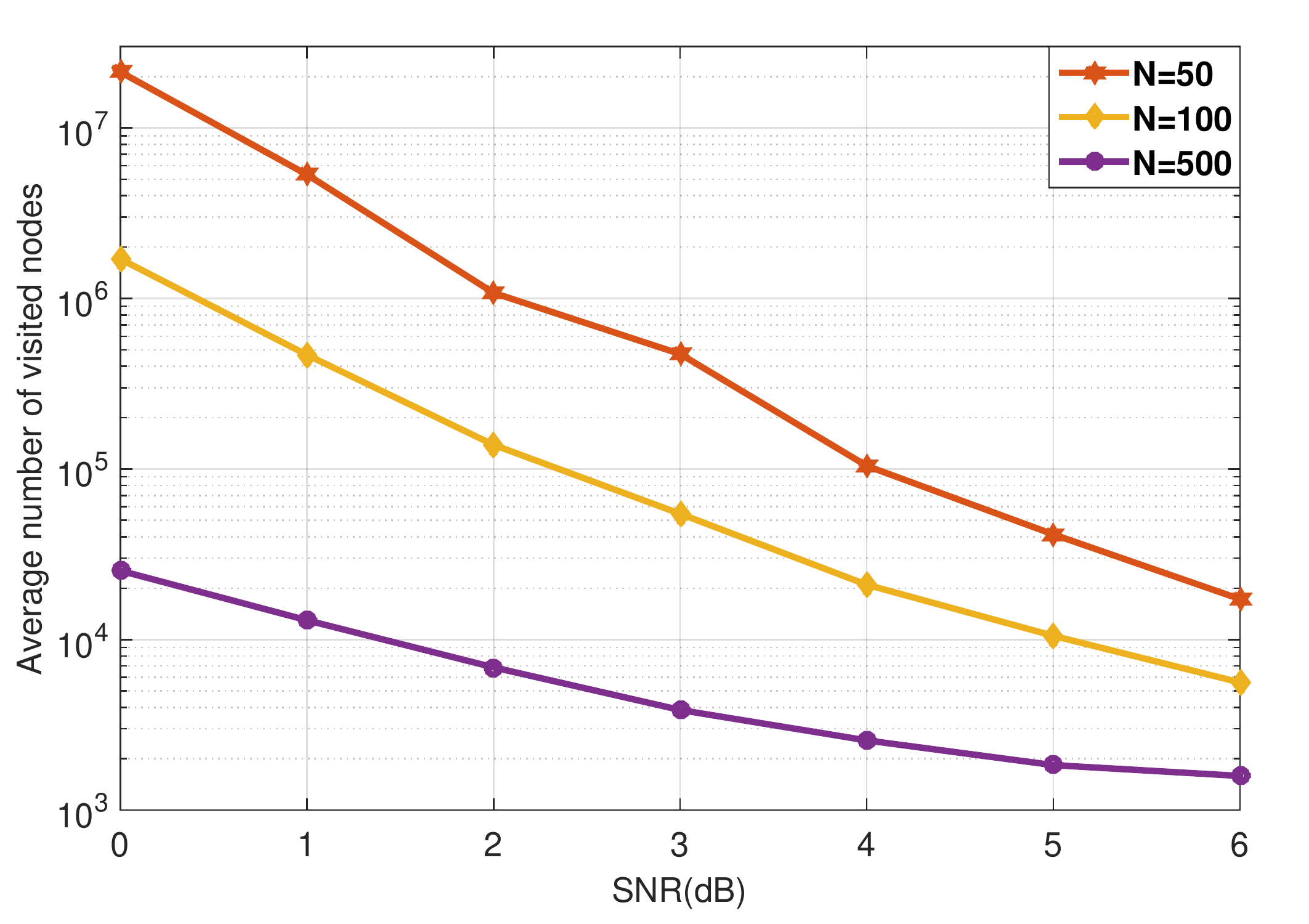}\\
  \caption{ Average number of visited points for $T=8$, and 16-QAM modulation. Exhaustive search will  instead need to test $2.8147\times 10^{14}$ hypotheses}\label{M2complexity}
\end{figure}
\begin{figure}[!t]
  \centering
 %\vspace{2.0cm}
  \includegraphics[width =2.5 in]{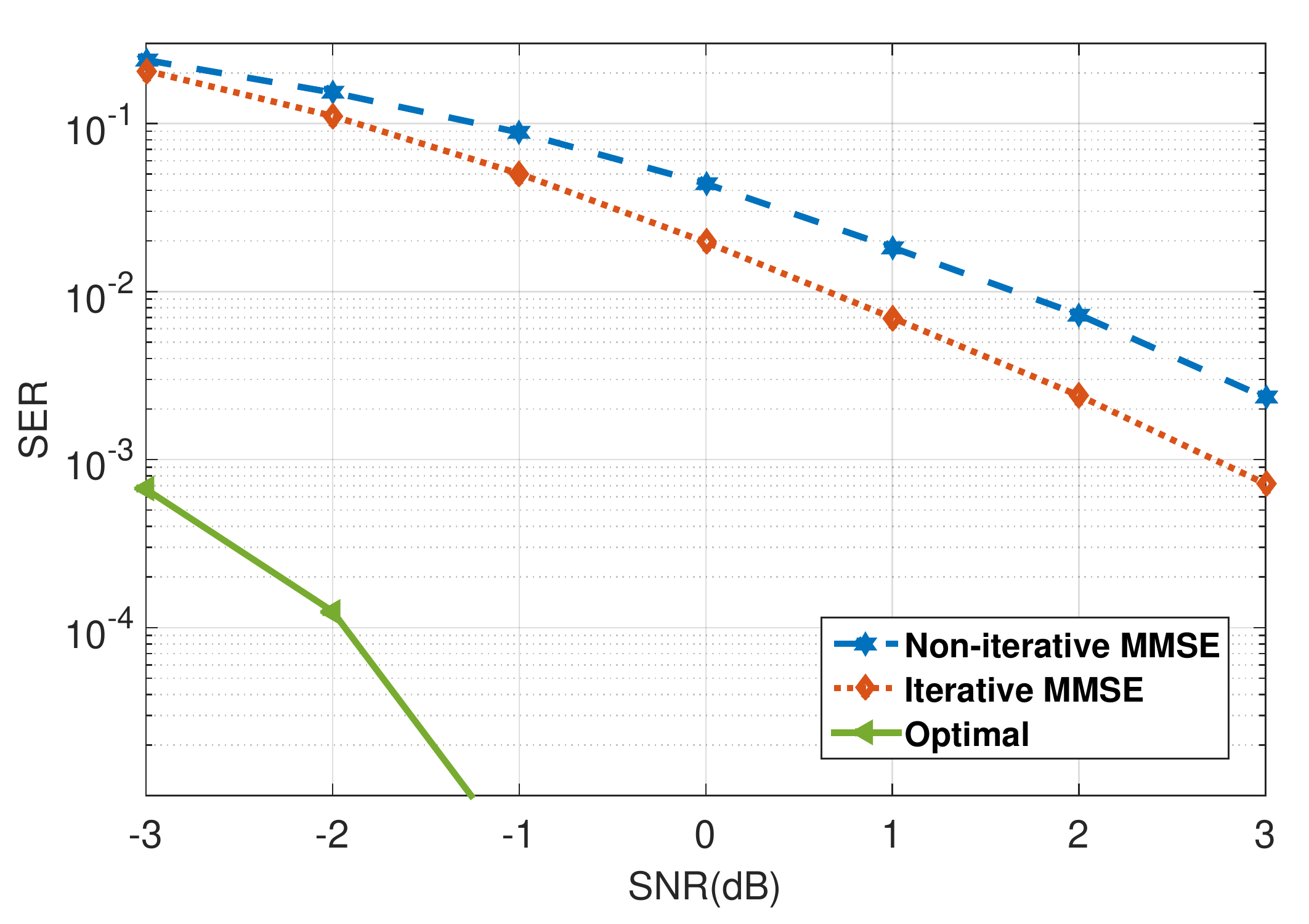}\\
  \caption{SER for iterative MMSE, non-iterative MMSE, and our optimal tree search algorithm. $M=2$, $T=8$, $N=200$, and 16-QAM modulation.}\label{M2N200Perform}
\end{figure}
\begin{figure}[!t]
  \centering
  \includegraphics[width =2.5 in]{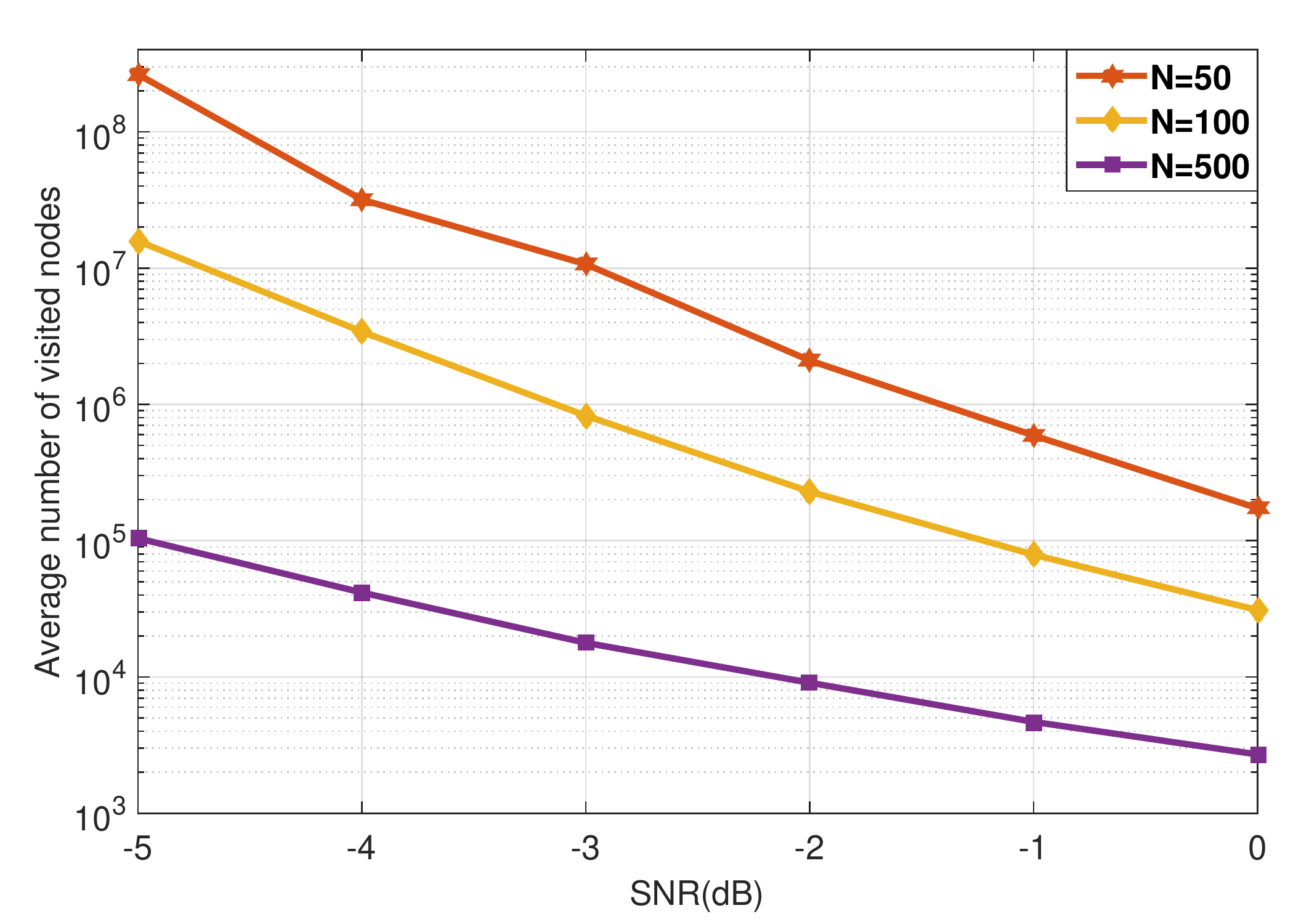}\\
  \caption{ Average number of visited points for $T=10$, $M=4$, and QPSK modulation. Exhaustive search will instead need to test  $2.8147\times 10^{14}$ hypotheses.}\label{M4complexity}
\end{figure}


\begin{thebibliography}{1}
\bibitem{Thomas}
T. L. Marzetta, ``Noncooperative cellular wireless with
unlimited numbers of base station antennas,'' \emph{IEEE Transaction on Wireless Comunication}, vol 9, no 11, pp.~3590-3600, Nov. 2003.
\bibitem{Improving_energy}
 E. Bjornson, M. Kountouris and M. Debbah, ``Massive MIMO and small cells: Improving energy efficiency by optimal soft-cell coordination," \emph{In 20th International Conference on Telecommunications (ICT)}, pp.1-5, 6-8 May 2013.
\bibitem{NextGeneration}
E. G. Larsson, O. Edfors, F. Tufvesson and T. L. Marzetta, ``Massive {MIMO} for next generation wireless systems," \emph{IEEE Communications Magazine}, vol. 52, pp. 186-195, Feb 2014.
\bibitem{MUMIMO}
H. Q. Ngo, E. G. Larsson and T. L. Marzetta, ``Massive {MU-MIMO} downlink {TDD} systems with linear precoding and downlink pilots," \emph{51st Annual Allerton Conference on Communication, Control, and Computing}, 2013, pp. ~293-298.
\bibitem{HowmanyAntenna}
 J. Hoydis, S. T. Brink and M. Debbah, ``Massive {MIMO}: How many antennas do we need? " \emph{49th Annual Allerton Conference on Communication, Control, and Computing}, 2011.
 \bibitem{Muller}
R.R. Muller, L. Cottatellucci and M. Vehkapera, ``Blind Pilot Decontamination," \emph{IEEE Journal of Selected Topics in Signal Processing}, vol.8, no.5, pp.773-786, Oct. 2014.
\bibitem{Ngo}
H. Q. Ngo and E. G. Larsson, ``EVD-based channel estimations for multicell multiuser MIMO with very large antenna arrays,” \emph{in IEEE International Conference on Acoustics, speech and signal processing (ICASSP)}, Mar. 2012.
 \bibitem{Wilson}
  S.G. Wilson, J. Freebersyser and C. Marshall, ``Multi-symbol detection of M-DPSK," \emph{in Global Telecommunications Conference and Exhibition 'Communications Technology for the 1990s and Beyond' (GLOBECOM)}, pp.1692-1697 vol.3, 27-30 Nov 1989.
\bibitem{Stoica}
P. Stoica and G. Ganesan, ``Space-time block codes: Trained, blind, and semi-blind detection,
 '' \emph{Digital Signal Process}, vol. 13, pp.~93-105, 2003.
\bibitem{Ma}
W-K. Ma, B-N. Vo, T. N. Davidson, and P-C. Ching, ``Blind {ML} detection of orthogonal space-time block codes: High-performance,
efficient implementations," \emph{IEEE Transactions on Signal Processing}, vol. 54, no. 2, pp.~738-751, Feb. 2006.
\bibitem{Hassibi}
H. Vikalo, B. Hassibi and P. Stoica, ``Efficient joint maximum-likelihood channel estimation and signal detection,"
\emph{IEEE Transactions on Wireless Communications}, vol 5, no 7, pp.~1838-1845, Jul 2006.
\bibitem{Optimal}
H. Alshamary, T. Al-Naffouri, A. Zaib, and W. Xu, ``Optimal non-coherent data detection for massive SIMO wireless systems: A polynomal complexity solution," \emph{Proceedings of IEEE Signal Processing and Signal Processing Education Workshop (SP/SPE)}, pp.172-177, 2015.
\bibitem{complexity}
B. Hassibi and  H. Vikalo, ``On the sphere-decoding algorithm I. Expected complexity," \emph{IEEE Transactions on Signal Processing }, vol.53, no.8, pp. 2806-2818, Aug. 2005.
\bibitem{complexity2}
J. Jald$\acute{e}$n and B. Ottersten, ``On the complexity of sphere decoding in digital communications," \emph{IEEE Transactions on Signal Processing }, vol.53, no.4, pp.1474-1484, April 2005
\bibitem{GLRT}
D.J. Ryan, I.B. Collings and I.V.L. Clarkson, ``GLRT-Optimal Noncoherent Lattice Decoding," \emph{IEEE Transactions on Signal Processing }, vol.55, no.7, pp.3773-3786, July 2007
\bibitem{Madhow}
D. Warrier and U. Madhow, ``Spectrally efficient noncoherent communication," IEEE Transactions on Information Theory, vol.48, no.3, pp.651-668, Mar 2002.
\bibitem{AAaproach}
D.S. Papailiopoulos, G.A. Elkheir and  G.N. Karystinos, ``Maximum-Likelihood Noncoherent PAM Detection," \emph{IEEE Transactions on Communications }, vol.61, no.3, pp.1152-1159, March 2013
\bibitem{Weiyu}
W. Xu, M. Stojnic and B. Hassibi, ``On exact maximum-likelihood detection for non-coherent MIMO wireless systems: A branch-estimate-bound optimization framework," \emph{IEEE International Symposium in Information Theory, 2008. ISIT}, pp.2017-2021, 6-11 July 2008.
\bibitem{Anastasopoulos}
I. Motedayen-Aval, A. Krishnamoorthy and A. Anastasopoulos, ``Optimal joint detection/estimation in fading channels with polynomial complexity," \emph{IEEE Trans. Inf. Theory}, 2007.
\bibitem{Jeon}
C. Jeon, R. Ghods, A. Maleki and C. Studer, ``Optimality of large MIMO detection via approximate message passing," \emph{in IEEE International Symposium on Information Theory (ISIT)}, pp.1227-1231, 14-19 June 2015.
\bibitem{Joint_channel}
W. Chao-Kai, J. Shi, W. Kai-Kit, W. Chang-Jen and W. Gang, ``Joint channel-and-data estimation for large-MIMO systems with low-precision ADCs," \emph{in IEEE International Symposium on Information Theory (ISIT)}, pp.1237-1241, 14-19 June 2015.
\end{thebibliography}
\end{document}